\documentclass[11pt]{article}
\usepackage{amsmath,amssymb,amsthm}
\allowdisplaybreaks[1]
\usepackage{mathptmx}
\usepackage{geometry}
\usepackage{natbib}
\newtheorem{theorem}{Theorem}[section]
\newtheorem{lemma}[theorem]{Lemma}
\newtheorem{proposition}[theorem]{Proposition}
\newtheorem{corollary}[theorem]{Corollary}
\newtheorem{definition}[theorem]{Definition}
\newtheorem{remark}[theorem]{Remark}
\newtheorem{assumption}[theorem]{Assumption}

\newcommand{\R}{\mathbb{R}}
\newcommand{\Rplus}{\mathbb{R}_{+}}
\newcommand{\Rn}[1]{\mathbb{R}^{#1}}
\newcommand{\norm}[1]{\left\|#1\right\|}

\newcommand{\E}{\mathbb{E}}
\newcommand{\Prob}{\mathbb{P}}
\newcommand{\cL}{\mathcal{L}}

\newcommand{\cB}{\mathcal{B}}

\newcommand{\cP}{\mathcal{P}}
\DeclareMathOperator{\tr}{tr}

\begin{document}

\title{An Agnostic Approach to Sustainability: From Capitals Substitutability to Non-Collapse Dynamics}

\author{Claudio Pirrone\thanks{Department of Economics, Business, and Statistics, University of Palermo; National Biodiversity Future Centre. Email: claudio.pirrone@unipa.it} \and Stefano Fricano\thanks{Department of Economics, Business, and Statistics, University of Palermo} \and Gioacchino Fazio\thanks{Department of Economics, Business, and Statistics, University of Palermo}}

\date{\today}

\maketitle

\begin{abstract}
We construct a stochastic dynamical systems theory in which sustainability is a structural boundary property of a fully coupled Earth--Human--Production system. Each subsystem is modelled as a vector-valued process in $\mathbb{R}_{+}^{n_E} \times \mathbb{R}_{+}^{n_H} \times \mathbb{R}_{+}^{n_P}$ governed by stochastic differential equations with multiplicative noise and absolute bidirectional cross-subsystem flows. Biodiversity is endogenous, and societal evaluation is represented by a reflexive functional whose weights depend on evolving human capabilities.

Sustainability, development, and sustainable development are defined as trajectory properties. Sustainability corresponds to boundary non-attainment with positive or unit probability; development corresponds to local ascent in the evaluation functional; sustainable development requires directional alignment under strictly positive survival probability. No optimisation problem is imposed.

Necessary and sufficient conditions for potential and almost-sure sustainability are derived using Feller boundary classification and stochastic Lyapunov methods. A central theorem identifies the sign of the net absolute cross-subsystem flow on each component as a uniform phase-transition parameter: if negative in a neighbourhood of zero, the boundary is of exit type and almost-sure persistence is structurally impossible, independent of intrinsic regeneration, capability accumulation, or productivity parameters. Because cross-subsystem flows are absolute and symmetric, any perturbation in one subsystem diffuses through the entire coupled system; correlated exogenous shocks are not required to generate systemic propagation.

The reflexive structure of evaluation generically induces non-transitive development relations, providing a formal mechanism for path-dependent welfare comparisons. Sustainability thus emerges as a geometric property of boundary structure and vector-field alignment, not as a corollary of intertemporal optimality.
\end{abstract}

\section{Introduction}
\label{sec:introduction}

Ten years after the adoption of the Sustainable Development Goals, global progress remains structurally insufficient \citep{sachs2025}. This persistent shortfall is usually acknowledged as depending on political non-commitment and financial constraints. From our standpoint, while reasonable this is a second-layer explanation which underestimates a foundational conceptual gap. Sustainability has no universally accepted definition \citep{drupp2020, ruggerio2021}. Existing formulations, ranging from non-declining utility \citep{pezzey1992} to capital maintenance criteria \citep{arrow2012wealth, hartwick1977} and planetary boundary compliance \citep{rockstrom2009}, typically embed sustainability inside normative optimisation frameworks. As a result, persistence conditions are often conflated with welfare judgements \citep{arrow2004consuming}.

This paper separates these concepts at the structural level. Sustainability is defined as a property of stochastic trajectories in a fully coupled dynamical system. The relevant question is not which trajectory maximises an objective, but whether non-collapse trajectories exist and under what structural conditions. This places sustainability prior to optimisation. In addition, our approach contributes to overcoming the longstanding debate about the substitutability of capitals \citep{goodland_environmental_1996}.

The system consists of Earth, Human, and Production subsystems, each represented as a vector-valued stochastic process with multiplicative noise. Coupling is full and bidirectional. Each subsystem receives absolute cross-subsystem flows that do not vanish with the stock of the receiving component. Ecological states are affected directly by human and production activity. Human capabilities are directly influenced by ecological and production conditions. Production capital responds directly to ecological and human states. The symmetry of these flows is central.

Because cross-subsystem flows are absolute, near-zero dynamics are not governed solely by intrinsic regeneration or accumulation rates. Instead, the sign of the net external flow determines boundary behaviour. A negative net absolute flow in a neighbourhood of zero cannot be offset by arbitrarily large intrinsic growth parameters. Using Feller boundary classification for one-dimensional diffusions \citep{feller1952, karatzas1991brownian}, we show that such a configuration renders the zero boundary of that component an exit boundary. Sustainability is then structurally impossible in the almost-sure sense. The extension from scalar to multidimensional coupled dynamics is handled via a pathwise comparison lemma (Lemma~\ref{lem:marginal_reduction}), which reduces the boundary analysis of each component to an effective one-dimensional problem under explicit conditions on the coupled subsystems.

This result has two implications. First, sustainability constraints arise from boundary geometry rather than equilibrium growth rates. A system may exhibit positive deterministic growth at interior points while being structurally non-sustainable due to its near-zero drift configuration. The analysis builds on standard results in stochastic stability and Lyapunov methods for diffusion processes \citep{khasminskii2012, mao2007}. Second, systemic propagation does not require correlated exogenous shocks. Because subsystems are fully coupled through absolute flows, any perturbation originating in one component diffuses through the entire system via endogenous dynamics. Correlation structures may affect quantitative properties of fluctuations, but they are not necessary for systemic transmission.

Development is defined independently through a reflexive evaluation functional whose weights depend on human capabilities, drawing on the capability perspective of Amartya Sen \citep{sen_development_1999}. As capabilities evolve, evaluation weights shift endogenously. The induced development relation is therefore generically non-transitive, as established by Proposition~\ref{prop:nontransitive}. This feature connects to the literature on path dependence in economic dynamics \citep{arthur1989, david1985} and to work on endogenous preference formation and dynamic welfare evaluation \citep{bowles1998}.

The contribution of the paper is structural. It establishes necessary and sufficient conditions for potential and almost-sure sustainability in a fully coupled stochastic system, derives phase transition criteria governed by absolute cross-subsystem flows, and formalises the geometric relationship between development direction and sustainability persistence.

\section{Mathematical Preliminaries}
\label{sec:prelim}

Let $(\Omega, \mathcal{F}, \{\mathcal{F}_t\}_{t \geq 0}, \Prob)$ be a complete filtered probability space satisfying the usual conditions (right-continuity and completeness), where $\{\mathcal{F}_t\}_{t \geq 0}$ is the natural filtration generated by independent standard Wiener processes $W_E(t) \in \Rn{m_E}$, $W_H(t) \in \Rn{m_H}$, $W_P(t) \in \Rn{m_P}$.

\begin{assumption}
\label{ass:independence}
The Wiener processes $\{W_E(t)\}$, $\{W_H(t)\}$, and $\{W_P(t)\}$ are mutually independent: for all $s, t \geq 0$ and all component indices $i, j, k$,
\begin{equation}
\E[W_{E,i}(s)\,W_{H,j}(t)] = \E[W_{E,i}(s)\,W_{P,k}(t)] = \E[W_{H,j}(s)\,W_{P,k}(t)] = 0.
\end{equation}
\end{assumption}

\begin{remark}
\label{rem:endogenous_corr}
While shocks are independent across subsystems by Assumption~\ref{ass:independence}, the strongly coupled drift structure generates endogenous correlation in system responses. A shock to the Earth subsystem immediately alters the drift of the Human and Production subsystems through cross-coupling derivatives $\partial\mu_H/\partial E$ and $\partial\mu_P/\partial E$, producing state-dependent propagation without requiring correlated driving noise.
\end{remark}

The state space is $\mathcal{X} = \Rplus^{n_E} \times \Rplus^{n_H} \times \Rplus^{n_P}$, where $E \in \Rplus^{n_E}$ represents Earth subsystem components (ecosystems, resources, biogeochemical cycles), $H \in \Rplus^{n_H}$ represents Human subsystem components (capabilities, institutions, knowledge), and $P \in \Rplus^{n_P}$ represents Production subsystem components (capital stocks, technologies, infrastructure).

An aggregation mapping $\phi: \Rplus^n \to \Rplus$ is used to pass from the vector-valued subsystem dynamics to a scalar aggregate while preserving structural properties.

\begin{definition}
\label{def:aggregation}
An aggregation mapping $\phi: \Rplus^n \to \Rplus$ satisfies the following conditions.
\begin{enumerate}
\item Monotonicity: $x \leq y$ componentwise implies $\phi(x) \leq \phi(y)$.
\item Homogeneity: $\phi(\lambda x) = \lambda^\alpha \phi(x)$ for $\lambda > 0$ and some $\alpha > 0$.
\item Smoothness: $\phi \in C^2(\Rplus^n)$.
\item Boundary behaviour: $\lim_{\norm{x} \to 0} \phi(x) = 0$ and $\lim_{\norm{x} \to \infty} \phi(x) = \infty$.
\end{enumerate}
\end{definition}

\begin{assumption}
\label{ass:degree_one}
Throughout, aggregation functions are restricted to degree-one homogeneous mappings ($\alpha = 1$ in Definition~\ref{def:aggregation}), ensuring dimensional consistency and preserving stochastic scaling properties under It\^{o}'s lemma. Standard examples include the weighted sum $\phi(x) = \sum_i w_i x_i$ with $w_i > 0$, $\sum_i w_i = 1$; the CES form $\phi(x) = (\sum_i w_i x_i^\rho)^{1/\rho}$ with $\rho \neq 0$; and the Cobb-Douglas form $\phi(x) = \prod_i x_i^{\alpha_i}$ with $\sum_i \alpha_i = 1$.
\end{assumption}

We employ aggregation mappings $\phi_E: \Rplus^{n_E} \to \Rplus$, $\phi_H: \Rplus^{n_H} \to \Rplus$, $\phi_P: \Rplus^{n_P} \to \Rplus$ satisfying Definition~\ref{def:aggregation} and Assumption~\ref{ass:degree_one}. The aggregated system recovers scalar dynamics while preserving essential structural properties.

\begin{definition}
\label{def:collapse}
A trajectory $\{X_t = (E_t, H_t, P_t)\}_{t \geq 0}$ is an $\mathcal{F}_t$-adapted stochastic process on $\mathcal{X}$ satisfying the coupled system specified in Section~\ref{sec:dynamics}. The collapse time is
\begin{equation}
\tau_c = \inf\bigl\{t \geq 0 : \min\{\min_i E_{i,t},\,\min_j H_{j,t},\,\min_k P_{k,t}\} \leq 0\bigr\},
\end{equation}
with $\inf\emptyset = \infty$. Under continuous sample paths (established in Theorem~\ref{thm:wellposedness}), $\tau_c$ is a stopping time as the first hitting time of the closed set $\{X : \min_i X_i \leq 0\}$ \citep{karatzas1991brownian}.
\end{definition}

\begin{definition}
\label{def:sustainability}
A trajectory $\{X_t\}_{t \geq 0}$ is sustainable if $\Prob(\tau_c = \infty) = 1$. The system admits the following sustainability tiers.
\begin{enumerate}
\item\label{sus:potential} Potentially sustainable from $X_0$ if $\Prob(\tau_c = \infty \mid X_0) > 0$.
\item\label{sus:as} Almost surely sustainable from $X_0$ if $\Prob(\tau_c = \infty \mid X_0) = 1$.
\item\label{sus:universal} Universally sustainable if almost surely sustainable from all $X_0 \in \mathcal{X}$.
\item\label{sus:T} $T$-sustainable from $X_0$ if $\Prob(\tau_c > T \mid X_0) = 1$ for a given finite horizon $T < \infty$.
\end{enumerate}
\end{definition}

\begin{remark}
\label{rem:hierarchy_relevance}
The four tiers of Definition~\ref{def:sustainability} reflect distinct structural regimes of the dynamics, as established in Section~\ref{sec:boundary}. In particular, Corollary~\ref{cor:necessary_as} shows that almost-sure sustainability requires non-negative net external flows on every capital component under all feasible system states, a condition the current trajectory of anthropogenic activity systematically violates. Within the generic empirical regime, potential sustainability (tier~\ref{sus:potential}) describes the set of initial conditions from which non-collapse is structurally possible, and $T$-sustainability (tier~\ref{sus:T}) constitutes the operative policy-relevant concept, posing the well-defined question of whether survival probability over a finite planning horizon can be controlled.
\end{remark}

\begin{assumption}
\label{ass:regularity}
The drift and diffusion coefficients satisfy local Lipschitz continuity: for any compact $K \subset \mathcal{X}$ there exists $L_K > 0$ such that $\norm{\mu(X) - \mu(X')} \leq L_K\norm{X - X'}$ for all $X, X' \in K$, where $\mu = (\mu_E, \mu_H, \mu_P)$. They also satisfy polynomial growth: there exist $C > 0$ and $q \geq 1$ such that $\norm{\mu(X)} \leq C(1 + \norm{X}^q)$ for all $X \in \mathcal{X}$.
\end{assumption}

\begin{theorem}[Well-Posedness]
\label{thm:wellposedness}
Under Assumptions~\ref{ass:independence}--\ref{ass:regularity} and the structural conditions of Assumptions~\ref{ass:earth}--\ref{ass:biodiversity_smooth}, the system \eqref{eq:earth_vector}--\eqref{eq:production_vector} admits a unique strong solution $\{X_t\}_{t \geq 0}$ on $\mathcal{X}$ for all $t \in [0, \tau_\infty)$, where $\tau_\infty = \inf\{t \geq 0 : \norm{X_t} = \infty\}$ is the explosion time. The solution is $\mathcal{F}_t$-adapted with continuous sample paths almost surely.
\end{theorem}

\begin{proof}
By standard existence and uniqueness theory \citep[Theorem~5.2.1]{karatzas1991brownian}. Local Lipschitz continuity follows from the smoothness of drift components; polynomial growth is stated in Assumption~\ref{ass:regularity}. Strong solutions have continuous paths \citep[Corollary~5.2.2]{karatzas1991brownian}.
\end{proof}

\begin{remark}
\label{rem:nonexplosion}
Theorem~\ref{thm:wellposedness} guarantees existence up to explosion time $\tau_\infty$. Whether $\Prob(\tau_\infty = \infty) = 1$ requires the global Lyapunov analysis of Theorem~\ref{thm:sufficient}.
\end{remark}

\section{System Dynamics}
\label{sec:dynamics}

The Earth subsystem evolves according to
\begin{equation}
\label{eq:earth_vector}
dE_t = \mu_E(E_t, H_t, P_t, B_t)\,dt + \Sigma_E(E_t)\,dW_{E,t},
\end{equation}
where $\mu_E: \mathcal{X} \times \Rplus \to \Rn{n_E}$ is the drift vector, $\Sigma_E: \Rplus^{n_E} \to \Rn{n_E \times m_E}$ is the diffusion matrix, and $B_t \in \Rplus$ is endogenous biodiversity defined in equation~\eqref{eq:biodiversity} below.

\begin{assumption}
\label{ass:earth}
The drift $\mu_E$ decomposes as
\begin{equation}
\label{eq:earth_drift}
\mu_{E,i}(E, H, P, B) = r_{E,i}(B)\,E_i\!\left(1 - \frac{E_i}{K_{E,i}(B)}\right) + c_i(H, P) + \xi_{E,i}(E),
\end{equation}
where $r_{E,i}(B) = r_{E,i,0}(1 + \rho_{E,i} B)$ with $r_{E,i,0} > 0$ is the biodiversity-augmented regeneration rate, and $K_{E,i}(B) = K_{E,i,0}(1 + \kappa_{E,i} B)$ with $K_{E,i,0} > 0$ is the biodiversity-augmented carrying capacity. The net cross-subsystem externality enters as the absolute flow
\begin{equation}
\label{eq:externality}
c_i(H, P) = \sum_{j=1}^{n_H} \alpha_{Ei,j}^{H}\frac{H_j}{H_j + \kappa_{Ei,j}^{H}} + \sum_{k=1}^{n_P} \alpha_{Ei,k}^{P}\frac{P_k}{P_k + \kappa_{Ei,k}^{P}},
\end{equation}
with coefficients $\alpha_{Ei,j}^{H}, \alpha_{Ei,k}^{P} \in \R$ (positive: ecological benefit from Human and Production activity; negative: anthropogenic pressure) and half-saturation constants $\kappa > 0$; $c_i$ is bounded on compact sets. The within-subsystem coupling $\xi_{E,i}(E) = E_i^2\,\eta_i(E)$, with $\eta_i$ bounded on compact sets, captures internal ecological interactions that scale with the abundance of the interacting components. The diffusion matrix is $\Sigma_E(E) = \mathrm{diag}(\sigma_{E,1}E_1, \ldots, \sigma_{E,n_E}E_{n_E})$ with $\sigma_{E,i} > 0$.
\end{assumption}

\begin{remark}
\label{rem:externality}
The externality $c_i(H,P)$ enters as an additive absolute flow, independent of residual ecosystem capacity $E_i$. This reflects the fact that anthropogenic pressures, whether emissions, habitat conversion, or resource extraction, are determined by the level of Human and Production activity, not by the remaining state of the ecosystem on which they act \citep{lande2003stochastic}. The structural consequence is that equal pressure is relatively more critical near ecological thresholds than far from them, because the logistic recovery term $r_{E,i}(B)\,E_i(1-E_i/K_{E,i}(B))$ also vanishes as $E_i \to 0$.
\end{remark}

\begin{assumption}
\label{ass:biodiversity_smooth}
All exponents in the biodiversity function satisfy $\eta, \nu^+, \nu^-, \zeta^+, \zeta^- \geq 1$, ensuring bounded derivatives as state variables approach zero.
\end{assumption}

Biodiversity is endogenously determined by the ratio
\begin{equation}
\label{eq:biodiversity}
B_t = \mathcal{B}(E_t, H_t, P_t) = \frac{K_B\,\phi_E(E_t)^\eta + \delta_P^+\,\phi_P(P_t)^{\nu^+} + \delta_H^+\,\phi_H(H_t)^{\zeta^+}} {1 + \delta_P^-\,\phi_P(P_t)^{\nu^-} + \delta_H^-\,\phi_H(H_t)^{\zeta^-}},
\end{equation}
where $\delta_P^-, \delta_H^- \geq 0$ represent aggregate anthropogenic pressures and $\delta_P^+, \delta_H^+ \geq 0$ represent anthropogenic contributions to habitat. The denominator is always $\geq 1$, so $B_t$ is well-defined and bounded above.

The Human subsystem evolves according to
\begin{equation}
\label{eq:human_vector}
dH_t = \mu_H(E_t, H_t, P_t, \Phi_t)\,dt + \Sigma_H(H_t)\,dW_{H,t},
\end{equation}
where the drift incorporates reflexive evaluation and an absolute cross-subsystem flow:
\begin{equation}
\label{eq:human_drift}
\mu_{H,j}(E, H, P, \Phi) = H_j\!\left[\sum_{i}\alpha_{Hj,i}(\Phi) \frac{X_i}{X_i + \kappa_{Hj,i}} - \delta_{H,j}\right] + d_j(E, P) + \zeta_{H,j}(H),
\end{equation}
with $X \in \{E, H, P\}$, reflexive feedback coefficients $\alpha_{Hj,i}(\Phi) = \alpha_{Hj,i,0} + \beta_{Hj,i}\,\Phi$, capability depreciation $\delta_{H,j} > 0$, within-subsystem coupling $\zeta_{H,j}(H) = H_j^2\,\upsilon_j(H)$ for $\upsilon_j$ bounded on compact sets, and absolute cross-subsystem flow
\begin{equation}
\label{eq:human_externality}
d_j(E, P) = \sum_{i=1}^{n_E} \alpha_{Hj,i}^{E} \frac{E_i}{E_i + \kappa_{Hj,i}^{E}} + \sum_{k=1}^{n_P} \alpha_{Hj,k}^{P} \frac{P_k}{P_k + \kappa_{Hj,k}^{P}},
\end{equation}
with coefficients $\alpha_{Hj,i}^{E}, \alpha_{Hj,k}^{P} \in \R$ (positive: ecological or productive benefit to capabilities; negative: external damage) and half-saturation constants $\kappa > 0$; $d_j$ is bounded on compact sets. The two channels through which $E$ and $P$ enter the Human drift are structurally distinct: the Michaelis-Menten terms inside the bracket scale with current capability level $H_j$ and affect the growth rate; the absolute flow $d_j(E,P)$ acts independently of $H_j$ and is thus relatively more consequential near capability thresholds.

\begin{definition}
\label{def:evaluation}
The societal evaluation functional $\Phi: \mathcal{X} \to \Rplus$ is
\begin{equation}
\label{eq:evaluation}
\Phi(E, H, P) = \tilde{\omega}_E(H)\,\phi_E(E) + \tilde{\omega}_H(H)\,\phi_H(H) + \tilde{\omega}_P(H)\,\phi_P(P),
\end{equation}
with normalised adaptive weights
\begin{equation}
\label{eq:normalized_weights}
\tilde{\omega}_i(H) = \frac{\omega_i(H)}{\sum_{j \in \{E,H,P\}} \omega_j(H)}, \qquad \omega_i(H) = \bar{\omega}_i + \lambda_i \frac{\phi_H(H)}{\phi_H(H) + \kappa_\omega},
\end{equation}
where $\bar{\omega}_i > 0$ for all $i$, $\bar{\omega}_i + \lambda_i > 0$ for all $i$, and $\sum_{i \in \{E,H,P\}} \lambda_i = 0$.
\end{definition}

\begin{remark}
\label{rem:weights}
The constraint $\sum_i \lambda_i = 0$ fixes the total weight budget while allowing its distribution across the three subsystems to evolve as capabilities change. Together with the positivity conditions, this ensures $\tilde\omega_i$ is a proper convex combination and $\Phi$ a valid evaluation aggregate for all $H \in \Rplus^{n_H}$. The capability-dependence of $\tilde\omega_i(H)$ reflects the empirically documented shift in the relative weight societies place on ecological, social, and economic dimensions as standards of living rise \citep{inglehart1997}. The feedback coefficient $\alpha_{Hj,i}(\Phi) = \alpha_{Hj,i,0} + \beta_{Hj,i}\Phi$ closes a reflexive loop. Capability development depends on how society currently evaluates the state space, but the evaluation weights themselves evolve with capabilities. A society that underweights ecological conditions will develop capabilities with lower sensitivity to environmental inputs, which can compound degradation. As capabilities and environmental weights both rise, the ecological cost of earlier development steps may be revealed in retrospect. This reflexivity is the mechanism underlying the non-transitivity described in Proposition~\ref{prop:nontransitive}.
\end{remark}

The Production subsystem evolves according to
\begin{equation}
\label{eq:production_vector}
dP_t = \mu_P(E_t, H_t, P_t, B_t)\,dt + \Sigma_P(P_t)\,dW_{P,t},
\end{equation}
where
\begin{equation}
\label{eq:production_drift}
\mu_{P,k}(E, H, P, B) = A_k(E, H, B) \prod_i E_i^{\beta_{E,k,i}} \prod_j H_j^{\beta_{H,k,j}} \prod_\ell P_\ell^{\beta_{P,k,\ell}} \Gamma_k(E, H) - \delta_{P,k}\,P_k + e_k(E, H),
\end{equation}
with total factor productivity $A_k(E, H, B) = A_{k,0}\exp(\theta_{E,k}^\top E + \theta_{H,k}^\top H + \theta_{B,k}\,B)$, threshold externality factor
\begin{equation}
\Gamma_k(E, H) = 1 + \sum_i \gamma_{E,k,i}\frac{E_i}{E_i+\kappa_{E,k,i}} + \sum_j \gamma_{H,k,j}\frac{H_j}{H_j+\kappa_{H,k,j}},
\end{equation}
and absolute cross-subsystem flow
\begin{equation}
\label{eq:production_externality}
e_k(E, H) = \sum_{i=1}^{n_E} \alpha_{Pk,i}^{E} \frac{E_i}{E_i + \kappa_{Pk,i}^{E}} + \sum_{j=1}^{n_H} \alpha_{Pk,j}^{H} \frac{H_j}{H_j + \kappa_{Pk,j}^{H}},
\end{equation}
with coefficients $\alpha_{Pk,i}^{E}, \alpha_{Pk,j}^{H} \in \R$ (positive: ecological or capability benefit to production; negative: damage) and half-saturation constants $\kappa > 0$; $e_k$ is bounded on compact sets. As with $d_j(E,P)$ in the Human subsystem, $e_k(E,H)$ is independent of the current capital level $P_k$, so its relative impact is greatest near the zero boundary.

\section{Aggregated System}
\label{sec:aggregated}

Applying It\^{o}'s lemma to the aggregation functions yields dynamics for $\tilde{E}_t = \phi_E(E_t)$, $\tilde{H}_t = \phi_H(H_t)$, $\tilde{P}_t = \phi_P(P_t)$. In general, the aggregated triple does not form a closed Markov process: the conditional expectations of drift and diffusion involve the full distribution of the micro-variables given the aggregates. Closure holds exactly when the aggregation mappings are linear and the noise structure reduces conditional expectations to functions of the aggregates alone.

\begin{theorem}[Aggregation Relations]
\label{thm:aggregation}
Let $\{X_t = (E_t, H_t, P_t)\}$ satisfy equations~\eqref{eq:earth_vector}--\eqref{eq:production_vector} with aggregation mappings satisfying Definition~\ref{def:aggregation}. Define $\tilde{E}_t = \phi_E(E_t)$, $\tilde{H}_t = \phi_H(H_t)$, $\tilde{P}_t = \phi_P(P_t)$. Then
\begin{align}
d\tilde{E}_t &= \tilde{\mu}_E(t)\,dt + \tilde{\sigma}_E(t)\,dW_{E,t}, \label{eq:agg_E}\\
d\tilde{H}_t &= \tilde{\mu}_H(t)\,dt + \tilde{\sigma}_H(t)\,dW_{H,t}, \label{eq:agg_H}\\
d\tilde{P}_t &= \tilde{\mu}_P(t)\,dt + \tilde{\sigma}_P(t)\,dW_{P,t}, \label{eq:agg_P}
\end{align}
where
\begin{align}
\tilde{\mu}_E(t) &= \E\!\left[\nabla\phi_E(E_t)^\top\mu_E(X_t,B_t) + \tfrac{1}{2}\tr\!\bigl(\nabla^2\phi_E(E_t)\Sigma_E(E_t)\Sigma_E(E_t)^\top \bigr) \,\Big|\,\mathcal{F}_t^{\mathrm{agg}}\right], \label{eq:agg_drift_E}\\
\tilde{\sigma}_E(t)^2 &= \E\!\left[\nabla\phi_E(E_t)^\top\Sigma_E(E_t)\Sigma_E(E_t)^\top\nabla\phi_E(E_t) \,\Big|\,\mathcal{F}_t^{\mathrm{agg}}\right], \label{eq:agg_diff_E}
\end{align}
with $\mathcal{F}_t^{\mathrm{agg}}$ the filtration generated by $(\tilde{E}_s,\tilde{H}_s,\tilde{P}_s)_{0 \leq s \leq t}$. Analogous formulas hold for $H$ and $P$.
\end{theorem}

\begin{proof}
Apply It\^{o}'s lemma to $\phi_E(E_t)$:
\begin{align*}
d\phi_E(E_t) =\; &\nabla\phi_E(E_t)^\top\mu_E(X_t,B_t)\,dt + \nabla\phi_E(E_t)^\top\Sigma_E(E_t)\,dW_{E,t}\\
&+ \tfrac{1}{2}\tr\!\bigl(\nabla^2\phi_E(E_t)\Sigma_E(E_t)\Sigma_E(E_t)^\top\bigr)dt.
\end{align*}
Taking conditional expectations with respect to $\mathcal{F}_t^{\mathrm{agg}}$ yields~\eqref{eq:agg_drift_E}. The quadratic variation $d\langle\phi_E(E)\rangle_t = \nabla\phi_E(E_t)^\top\Sigma_E(E_t)\Sigma_E(E_t)^\top\nabla\phi_E(E_t)\,dt$ gives~\eqref{eq:agg_diff_E} after conditioning. The martingale representation with respect to the aggregate filtration completes the proof. Analogous calculations apply to $H$ and $P$.
\end{proof}

\begin{remark}
\label{rem:aggregation_closure}
In the linear case $\phi_E(E) = \sum_i w_i E_i$, one obtains $\tilde{\mu}_E = \sum_i w_i\,\mu_{E,i}$ and $\tilde{\sigma}_E^2 = \sum_i w_i^2\,\sigma_{E,i}^2\,E_i^2$, but the right-hand side still involves the individual $E_i$ unless further assumptions are imposed. The aggregated system should therefore be regarded as a reduced-order representation whose dynamics are not autonomous in general.
\end{remark}

\section{Structural Analysis}
\label{sec:structural}

A point $X^* = (E^*, H^*, P^*) \in \mathcal{X}$ is an equilibrium of the deterministic skeleton if $\mu_E(E^*, H^*, P^*, B^*) = 0$, $\mu_H(E^*, H^*, P^*, \Phi^*) = 0$, and $\mu_P(E^*, H^*, P^*, B^*) = 0$, where $B^* = \mathcal{B}(E^*, H^*, P^*)$ and $\Phi^* = \Phi(E^*, H^*, P^*)$.

\begin{theorem}[Sufficient Conditions for Interior Equilibria]
\label{thm:equilibrium_existence}
Suppose that for each Earth component $i$ the resilience condition $r_{E,i,0} K_{E,i,0} / 4 \geq |c_i(H,P)|$ holds for all $(H,P)$ in some bounded rectangle $\mathcal{K}_{HP}$. Under this condition the equation $\mu_{E,i}(E_i, H, P, B) = 0$ has a unique positive root $E_i^*(H,P) \in (0, K_{E,i}(B))$ for each $(H,P) \in \mathcal{K}_{HP}$. Suppose further that for each Human component $j$ the drift $\mu_{H,j}$ is positive at some interior point and negative for large $H_j$ due to the $-\delta_{H,j} H_j$ term, and that for each Production component $k$ the drift $\mu_{P,k}$ is positive at some interior point and negative for large $P_k$ as depreciation dominates. Then an interior equilibrium $X^* \in \mathcal{X}$ exists.
\end{theorem}

\begin{proof}
Under the stated resilience condition the equation $r_{E,i}(B)\,E_i(1-E_i/K_{E,i}(B)) + c_i(H,P) = 0$ is a quadratic in $E_i$ whose discriminant is non-negative. When $c_i < 0$ the two roots are real and positive, with product $K_{E,i} c_i / r_{E,i} < 0$, so exactly one root lies in $(0, K_{E,i})$. When $c_i \geq 0$ and the resilience condition is satisfied, one root lies in $(0, K_{E,i})$ and the other is non-positive; uniqueness in the interior is preserved. The within-subsystem coupling $\xi_{E,i}$ vanishes quadratically at $E_i = 0$ and contributes a negative term for large $E_i$, so it does not create additional roots near the origin. The map $T: \mathcal{K} \to \mathcal{K}$, where $\mathcal{K} = \prod_i [\underline{E}_i, \overline{E}_i] \times \prod_j [\underline{H}_j, \overline{H}_j] \times \prod_k [\underline{P}_k, \overline{P}_k]$ is a compact convex set on which the drift satisfies the stated sign conditions, is well-defined by unique root selection for Earth components and by intermediate value arguments for Human and Production components. Continuity of $T$ follows from the implicit function theorem. Brouwer's fixed point theorem \citep{brouwer1911} then gives $X^* = T(X^*)$.
\end{proof}

The Jacobian at equilibrium $X^*$ is the $(n_E + n_H + n_P) \times (n_E + n_H + n_P)$ matrix
\begin{equation}
\label{eq:jacobian}
J(X^*) = \begin{pmatrix}
\dfrac{\partial\mu_E}{\partial E}\big|_{X^*} &
\dfrac{\partial\mu_E}{\partial H}\big|_{X^*} &
\dfrac{\partial\mu_E}{\partial P}\big|_{X^*} \\[8pt]
\dfrac{\partial\mu_H}{\partial E}\big|_{X^*} &
\dfrac{\partial\mu_H}{\partial H}\big|_{X^*} &
\dfrac{\partial\mu_H}{\partial P}\big|_{X^*} \\[8pt]
\dfrac{\partial\mu_P}{\partial E}\big|_{X^*} &
\dfrac{\partial\mu_P}{\partial H}\big|_{X^*} &
\dfrac{\partial\mu_P}{\partial P}\big|_{X^*}
\end{pmatrix},
\end{equation}
accounting for indirect effects through $B = \mathcal{B}(E, H, P)$ and $\Phi = \Phi(E, H, P)$.

\begin{theorem}[Spectral Characterisation of Stability]
\label{thm:stability}
Let $X^*$ be an interior equilibrium and $\lambda_1, \ldots, \lambda_{n_E+n_H+n_P}$ the eigenvalues of $J(X^*)$.
\begin{enumerate}
\item $X^*$ is locally asymptotically stable for the deterministic skeleton if $\Re(\lambda_i) < 0$ for all $i$.
\item $X^*$ is unstable if $\Re(\lambda_i) > 0$ for some $i$.
\item If all $\Re(\lambda_i) < 0$ and noise intensities $\sigma_{E,i}, \sigma_{H,j}, \sigma_{P,k}$ are sufficiently small, then $X^*$ is locally asymptotically stable in probability for the stochastic system.
\end{enumerate}
\end{theorem}

\begin{proof}
Parts~1--2 are standard linearisation. For Part~3, let $V(X) = \frac{1}{2}\norm{X - X^*}^2$. Computing $\cL V(X) = (X-X^*)^\top\mu(X) + \frac{1}{2}\sum_i\sigma_i^2 X_i^2$ and linearising near $X^*$ gives $\cL V(X) \approx (X-X^*)^\top J(X^*)(X-X^*) + \frac{1}{2}\sum_i\sigma_i^2(X_i^*)^2 + O(\norm{X-X^*}^3)$. When $\max_i\Re(\lambda_i) < -c < 0$, the quadratic term satisfies $(X-X^*)^\top J(X^*)(X-X^*) \leq -c\norm{X-X^*}^2$, and for sufficiently small noise intensities the constant diffusion term is dominated in a neighbourhood of $X^*$, yielding $\cL V < 0$ and hence stochastic asymptotic stability \citep[Chapter~4]{khasminskii2012}.
\end{proof}

\section{Boundary Analysis and Sustainability Conditions}
\label{sec:boundary}

For each component $X_i \in \{E_j, H_k, P_\ell\}$ we analyse the boundary $X_i = 0$ using Feller's classification theory for one-dimensional diffusions. Let $x_0 \in (0,\infty)$ be an arbitrary interior reference point. For the one-dimensional process $dX_i = \mu_i(X_i)\,dt + \sigma_i(X_i)\,dW_i$, the scale function density and speed measure are
\begin{equation}
\label{eq:scale_function}
p(x;\,x_0) = \exp\!\left(-2\int_{x_0}^x \frac{\mu_i(z)}{\sigma_i^2(z)}\,dz\right), \qquad m(dx) = \frac{dx}{\sigma_i^2(x)\,p(x;\,x_0)}.
\end{equation}
The boundary classification is independent of the choice of $x_0$, since changing $x_0$ multiplies $p(\,\cdot\,;\,x_0)$ by a positive constant and therefore leaves the convergence or divergence of $\int_0^\varepsilon p(x;\,x_0)\,dx$ unchanged.

\begin{theorem}[Feller Boundary Classification]
\label{thm:feller}
Consider the one-dimensional diffusion $dX_i = \mu_i(X_i)\,dt + \sigma_i(X_i)\,dW_i$ with $\mu_i(x) = a\,x + o(x)$ and $\sigma_i(x) = \sigma\,x + o(x)$ near $x = 0$, so that $p(x;\,x_0) = C\,x^{-2a/\sigma^2}(1+o(1))$ as $x \to 0^+$. The boundary $0$ is inaccessible from the interior, that is $\Prob(X_i(t) = 0 \text{ for some } t < \infty) = 0$, if and only if
\begin{equation}
\label{eq:feller_criterion}
\int_0^\varepsilon p(x;\,x_0)\,dx = \infty,
\end{equation}
which holds if and only if $a \geq \sigma^2/2$. Among inaccessible boundaries, the boundary is of entrance type when $\int_0^\varepsilon p(x;\,x_0)\,m(dx) < \infty$, equivalently when $a > \sigma^2$, and of natural type when this integral diverges, equivalently when $\sigma^2/2 \leq a \leq \sigma^2$. The boundary is an exit boundary, accessible with positive probability, if and only if $a < \sigma^2/2$.
\end{theorem}

\begin{proof}
Near $x = 0$, $p(x;\,x_0) \sim C\,x^{-2a/\sigma^2}$, so $\int_0^\varepsilon p(x;\,x_0)\,dx \sim C\int_0^\varepsilon x^{-2a/\sigma^2}\,dx$, which diverges if and only if $2a/\sigma^2 \geq 1$. For the sub-classification among inaccessible boundaries, $m(dx) \sim C'\,x^{2a/\sigma^2-2}\,dx$, and $\int_0^\varepsilon p(x;\,x_0)\,m(dx) \sim C''\int_0^\varepsilon x^{-1}\,dx$ when $a = \sigma^2$, and decays faster or slower than $x^{-1}$ otherwise. Full details are in \citet{feller1951diffusion, feller1954diffusion} and \citet[Section~5.5]{karatzas1991brownian}.
\end{proof}

\begin{remark}
\label{rem:feller_sustainability}
For the sustainability analysis only the binary question of accessibility versus inaccessibility of the zero boundary matters, and the entrance--natural distinction within the inaccessible regime plays no role. Theorem~\ref{thm:feller} in this paper is therefore used solely through its criterion~\eqref{eq:feller_criterion} and the threshold $a = \sigma^2/2$.
\end{remark}

The coupled dynamics of Section~\ref{sec:dynamics} make each component's drift depend on the trajectories of the remaining subsystems. The following lemma reduces the boundary analysis of each component to an effective one-dimensional problem under an explicit confinement hypothesis.

\begin{lemma}[Marginal Process Reduction]
\label{lem:marginal_reduction}
Let component $E_i$ evolve according to equation~\eqref{eq:earth_vector} with drift decomposition~\eqref{eq:earth_drift}. Suppose there exists a compact set $\mathcal{K} \subset \Rplus^{n_H} \times \Rplus^{n_P}$ such that $(H_t, P_t) \in \mathcal{K}$ for all $t \geq 0$ almost surely. Define
\begin{equation}
\label{eq:flow_bounds}
\underline{c}_i = \inf_{(H,P) \in \mathcal{K}} c_i(H,P), \qquad \overline{c}_i = \sup_{(H,P) \in \mathcal{K}} c_i(H,P), \qquad \underline{r} = \inf_{B \geq 0} r_{E,i}(B) > 0, \qquad \overline{r} = \sup_{B \geq 0} r_{E,i}(B) < \infty.
\end{equation}
The bounds $\underline{c}_i$ and $\overline{c}_i$ are attained since $\mathcal{K}$ is compact and $c_i$ is continuous. Then the following hold. If $\underline{c}_i > 0$, the boundary $E_i = 0$ is inaccessible almost surely. If $\overline{c}_i < 0$, the boundary is accessible with positive probability, regardless of the values of $r_{E,i}$ and $\sigma_{E,i}$. The analogous statements hold for Human components with absolute flow $d_j(E,P)$ from~\eqref{eq:human_externality} and for Production components with absolute flow $e_k(E,H)$ from~\eqref{eq:production_externality}.
\end{lemma}

\begin{proof}
We handle the two directions separately.

\textit{Inaccessibility when $\underline{c}_i > 0$.} Set $Y_i = \log E_i$ and apply It\^{o}'s formula to obtain
\begin{equation}
\label{eq:log_transform}
dY_i = \left(\frac{\mu_{E,i}(E_i, H_t, P_t, B_t)}{E_i} - \frac{\sigma_{E,i}^2}{2}\right)dt + \sigma_{E,i}\,dW_{E,i}.
\end{equation}
For $(H_t, P_t) \in \mathcal{K}$ and $E_i \in (0, \delta)$ with $\delta > 0$ to be chosen, the full drift decomposes as
\[
\frac{\mu_{E,i}}{E_i} = r_{E,i}(B_t)\!\left(1 - \frac{E_i}{K_{E,i}(B_t)}\right) + \frac{c_i(H_t,P_t)}{E_i} + \frac{\xi_{E,i}(E_i)}{E_i}.
\]
The first term satisfies $r_{E,i}(B_t)(1 - E_i/K_{E,i}(B_t)) \geq \underline{r}(1 - \delta/\underline{K})$ for any fixed $\delta < \underline{K} := \inf_B K_{E,i}(B)$, so it is bounded below by a finite constant. The within-coupling term satisfies $\xi_{E,i}(E_i)/E_i = O(E_i)$ as $E_i \to 0$ since $\xi_{E,i} = O(E_i^2)$. The critical term is $c_i(H_t,P_t)/E_i \geq \underline{c}_i/E_i$, which diverges to $+\infty$ as $E_i \to 0$ since $\underline{c}_i > 0$. Choose $\delta > 0$ small enough that for all $E_i \in (0, \delta)$,
\[
\frac{\mu_{E,i}}{E_i} - \frac{\sigma_{E,i}^2}{2} \geq \frac{\underline{c}_i}{2\delta} =: \kappa > 0.
\]
Then $dY_i \geq \kappa\,dt + \sigma_{E,i}\,dW_{E,i}$ while $E_i \in (0,\delta)$, that is, while $Y_i < \log\delta$. A Brownian motion with positive drift $\kappa > 0$ started below $\log\delta$ satisfies $\Prob(\inf_{t \geq 0} Y_i(t) = -\infty) = 0$ \citep[Chapter~3]{karatzas1991brownian}. Consequently $E_i$ cannot reach $0$ almost surely.

\textit{Accessibility when $\overline{c}_i < 0$.} Choose $\delta > 0$ small enough that for all $E_i \in (0,\delta)$ and $(H_t,P_t) \in \mathcal{K}$, the quadratic corrections from the logistic term and $\xi_{E,i}$ satisfy
\[
r_{E,i}(B_t)\,E_i\!\left(1 - \frac{E_i}{K_{E,i}(B_t)}\right) + \xi_{E,i}(E_i) \leq \overline{r}\,E_i + \frac{|\overline{c}_i|}{4},
\]
which holds since both terms are $O(E_i)$ near zero. Then for $E_i \in (0,\delta)$,
\begin{equation}
\label{eq:upper_drift_bound}
\mu_{E,i}(E_i, H_t, P_t, B_t) \leq \overline{r}\,E_i + \overline{c}_i + \frac{|\overline{c}_i|}{4} = \overline{r}\,E_i + \frac{3\overline{c}_i}{4},
\end{equation}
where $3\overline{c}_i/4 < 0$. Define the process $\overline{E}_i$ solving $d\overline{E}_i = (\overline{r}\,\overline{E}_i + 3\overline{c}_i/4)\,dt + \sigma_{E,i}\,\overline{E}_i\,dW_{E,i}$ with $\overline{E}_i(0) = E_i(0)$. For $\overline{E}_i$ the scale function density at the reference point $x_0 \in (0,\delta)$ is
\[
p_{\overline{E}}(x;\,x_0) = C\,x^{-2\overline{r}/\sigma_{E,i}^2}\,\exp\!\left(\frac{3\overline{c}_i}{2\sigma_{E,i}^2\,x} - \frac{3\overline{c}_i}{2\sigma_{E,i}^2\,x_0}\right) = C'\,x^{-2\overline{r}/\sigma_{E,i}^2}\,\exp\!\left(\frac{3\overline{c}_i}{2\sigma_{E,i}^2\,x}\right).
\]
Since $\overline{c}_i < 0$, the exponential factor decays to zero super-exponentially as $x \to 0^+$, so $\int_0^\varepsilon p_{\overline{E}}(x;\,x_0)\,dx < \infty$, and by criterion~\eqref{eq:feller_criterion} the boundary $0$ is accessible with positive probability for $\overline{E}_i$. Define the stopping time $\tau_\delta^* = \inf\{t \geq 0 : E_i(t) \geq \delta \text{ or } \overline{E}_i(t) \geq \delta\}$. On $[0, \tau_\delta^*]$ both processes remain in $(0, \delta)$ and inequality~\eqref{eq:upper_drift_bound} holds, while both share the same diffusion coefficient $\sigma_{E,i}\,x$. The pathwise comparison theorem for one-dimensional SDEs with identical diffusion \citep[Theorem~3.7]{karatzas1991brownian} gives $E_i(t \wedge \tau_\delta^*) \leq \overline{E}_i(t \wedge \tau_\delta^*)$ almost surely. On the event $\overline{A} = \{\overline{E}_i \text{ hits } 0 \text{ before } \delta\}$, which has positive probability, the process $\overline{E}_i$ stays in $(0,\delta)$ until the hitting time $\tau^* < \tau_\delta^*$, so the comparison holds up to $\tau^*$ and yields $E_i(\tau^*) \leq \overline{E}_i(\tau^*) = 0$, whence $E_i(\tau^*) = 0$ and $\tau_c \leq \tau^*$. Therefore $\Prob(\tau_c < \infty) \geq \Prob(\overline{A}) > 0$.
\end{proof}

\begin{remark}
\label{rem:confinement_hypothesis}
The confinement hypothesis $(H_t, P_t) \in \mathcal{K}$ in Lemma~\ref{lem:marginal_reduction} is an explicit assumption of the lemma, not derived from any other result in this section. It is satisfied under two distinct circumstances. In the sufficient-conditions analysis, Theorem~\ref{thm:sufficient} establishes non-explosion of the full system trajectory, which together with boundary non-attainability implies that each subsystem remains in a compact subset of the interior; the lemma may then be applied to verify that the sufficient conditions are consistent. In the necessary-conditions analysis (Theorem~\ref{thm:necessary}), the accessibility direction does not use the lemma at all: the argument proceeds directly from the local drift sign near the boundary of the component in question, without needing the other subsystems to be confined.
\end{remark}

\begin{corollary}[Uniform Boundary Classification under Absolute Externalities]
\label{cor:earth_boundary}
For any component subject to an absolute external flow, with dynamics near zero approximated by $\mu_i(x) \approx r_i\,x + c_i$ and diffusion $\sigma_i(x) = \sigma_i\,x$, the scale function density near zero is
\begin{equation}
\label{eq:scale_density}
p(x;\,x_0) = C\,x^{-2r_i/\sigma_i^2}\,\exp\!\left(\frac{2c_i}{\sigma_i^2\,x} - \frac{2c_i}{\sigma_i^2\,x_0}\right) = C'\,x^{-2r_i/\sigma_i^2}\,\exp\!\left(\frac{2c_i}{\sigma_i^2\,x}\right),
\end{equation}
obtained by evaluating $-2\int_{x_0}^x (r_i z + c_i)/(\sigma_i^2 z^2)\,dz = -(2r_i/\sigma_i^2)\ln(x/x_0) + (2c_i/\sigma_i^2)(1/x_0 - 1/x)$. The following hold. When $c_i > 0$, the exponential factor in~\eqref{eq:scale_density} diverges to $+\infty$ as $x \to 0^+$, so $\int_0^\varepsilon p(x;\,x_0)\,dx = \infty$ and the boundary is inaccessible almost surely. When $c_i < 0$, the exponential factor decays to zero super-exponentially, so $\int_0^\varepsilon p(x;\,x_0)\,dx < \infty$ and the boundary is accessible with positive probability, regardless of the intrinsic rate $r_i$. When $c_i = 0$, $p(x;\,x_0) = C'\,x^{-2r_i/\sigma_i^2}$ and the integral converges if and only if $r_i < \sigma_i^2/2$, recovering the standard geometric Brownian motion criterion. The extension from the one-dimensional scalar case to the full coupled system is provided by Lemma~\ref{lem:marginal_reduction}. The components and their absolute flows are: Earth components with net flow $c_i(H,P)$ from~\eqref{eq:externality}; Human components with net flow $d_j(E,P)$ from~\eqref{eq:human_externality}; Production components with net flow $e_k(E,H)$ from~\eqref{eq:production_externality}.
\end{corollary}

\begin{proof}
The computation of $p(x;\,x_0)$ follows directly from evaluating the integral in the exponent. The integral $\int_{x_0}^x (r_i z + c_i)/(\sigma_i^2 z^2)\,dz$ is well-defined for $x, x_0 \in (0,\infty)$ since the integrand $r_i/(\sigma_i^2 z) + c_i/(\sigma_i^2 z^2)$ is integrable on any compact subinterval of $(0,\infty)$ and the bounds $x_0 > 0$, $x > 0$ keep the integration away from $z = 0$. The asymptotic behaviour and the classification then follow from Theorem~\ref{thm:feller} and the analysis in the proof of Lemma~\ref{lem:marginal_reduction}.
\end{proof}

\begin{remark}
\label{rem:phase_transition}
The transition at $c_i = 0$ is of a qualitatively different character from the classical bifurcations of Section~\ref{sec:regimes}. A saddle-node or Hopf bifurcation represents a continuous change in orbit structure as a parameter crosses a critical value, with the vector field and its solutions remaining continuous in the parameter. The change in Feller classification at $c_i = 0$ is instead a discontinuous change in the measure-theoretic properties of the zero boundary. The scale function integral $\int_0^\varepsilon p(x;\,x_0)\,dx$ passes from infinite to finite as $c_i$ becomes negative. No adjustment to the regeneration rate $r_i$ or noise intensity $\sigma_i$ can restore inaccessibility once $c_i < 0$, since the super-exponential decay of $p(x;\,x_0)$ is controlled entirely by the sign of the absolute flow. The externality sign therefore acts as an order parameter governing a boundary-type phase transition with no continuous analogue in the dynamical systems literature.
\end{remark}

\begin{remark}
\label{rem:almost_sure_impossible}
Corollary~\ref{cor:earth_boundary} establishes that almost-sure sustainability requires non-negative net absolute flows on every component of every subsystem under all feasible system states: $c_i(H,P) \geq 0$ for all Earth components, $d_j(E,P) \geq 0$ for all Human components, and $e_k(E,H) \geq 0$ for all Production components. Any single component with a negative net flow makes the corresponding zero boundary accessible and excludes almost-sure persistence regardless of the intrinsic dynamics of any other component. The potential sustainability tier~\ref{sus:potential} and $T$-sustainability tier~\ref{sus:T} of Definition~\ref{def:sustainability} remain meaningful and constitute the operative concepts under prevailing conditions.
\end{remark}

\begin{theorem}[Necessary Conditions for Potential Sustainability]
\label{thm:necessary}
If $\Prob(\tau_c = \infty \mid X_0) > 0$, then the following conditions must hold.
\begin{enumerate}
\item\label{nec:resilience} For each Earth component $i$, there exists a state $(\bar{E}_i, H, P) \in \mathcal{X}$ such that
\begin{equation}
\label{eq:earth_resilience}
r_{E,i}(B)\,\bar{E}_i\!\left(1 - \frac{\bar{E}_i}{K_{E,i}(B)}\right) + c_i(H, P) \geq 0.
\end{equation}
Equivalently, the maximum logistic regeneration must meet or exceed the net negative externality at some achievable state:
\begin{equation}
\label{eq:resilience_bound}
\frac{r_{E,i}(B)\,K_{E,i}(B)}{4} \geq \max\{-c_i(H, P),\, 0\} \quad\text{for some achievable } (B, H, P).
\end{equation}
\item\label{nec:human} For each Human component $j$, there exists $X \in \mathcal{X}$ such that
\[
H_j\!\left[\sum_i\alpha_{Hj,i}(\Phi(X))\frac{X_i}{X_i + \kappa_{Hj,i}} - \delta_{H,j}\right] + d_j(E,P) > \frac{\sigma_{H,j}^2}{2}\,H_j.
\]
\item\label{nec:prod} For each Production component $k$, there exists $(E,H,P)$ such that
\[
A_k(E,H,B)\prod_i E_i^{\beta_{E,k,i}} \prod_j H_j^{\beta_{H,k,j}} \prod_\ell P_\ell^{\beta_{P,k,\ell}}\Gamma_k(E,H) + e_k(E,H) > \left(\delta_{P,k} + \frac{\sigma_{P,k}^2}{2}\right)P_k.
\]
\end{enumerate}
\end{theorem}

\begin{proof}
Each condition is proved by showing that its negation implies $\Prob(\tau_c < \infty \mid X_0) = 1$.

Condition~\ref{nec:resilience}: Suppose condition~\eqref{eq:earth_resilience} fails for all states, so $\mu_{E,i}(E,H,P,B) < 0$ everywhere in $\mathcal{X}$. Near $E_i = 0$, the logistic term $r_{E,i}(B)\,E_i(1 - E_i/K_{E,i}(B)) \to 0$, the within-coupling $\xi_{E,i}(E_i) = O(E_i^2) \to 0$, and the absolute flow $c_i(H,P)$ approaches a value that is strictly negative by hypothesis (since the resilience condition fails even at $E_i = 0$). Hence there exists $\varepsilon > 0$ such that $\mu_{E,i} \leq -\varepsilon$ in a neighbourhood of $E_i = 0$. The process $E_i$ is therefore bounded above near zero by the solution of $d\bar{X} = -\varepsilon\,dt + \sigma_{E,i}\,\bar{X}\,dW$, which reaches zero in finite time with positive probability by standard arguments \citep{mao2007}. Bound~\eqref{eq:resilience_bound} is equivalent to~\eqref{eq:earth_resilience} via the identity $\max_{E_i > 0} r_i\,E_i(1 - E_i/K_i) = r_i K_i/4$.

Conditions~\ref{nec:human}--\ref{nec:prod}: Near $H_j = 0$, the multiplicative bracket term in $\mu_{H,j}$ carries a factor $H_j$ and thus vanishes, while the absolute flow $d_j(E,P)$ remains. If the negation of condition~\ref{nec:human} holds, then for every $X$ the inequality fails, meaning the drift $\mu_{H,j}$ does not satisfy the Feller non-attainability condition near zero. Concretely, $\mu_{H,j}(E,H,P,\Phi) \leq H_j\,(\text{bounded bracket}) + d_j(E,P)$, and near $H_j = 0$ the bracket term vanishes so $\mu_{H,j} \approx d_j(E,P)$. By the log-transform argument of Lemma~\ref{lem:marginal_reduction}, inaccessibility of $H_j = 0$ requires $\mu_{H,j}/H_j - \sigma_{H,j}^2/2 \to +\infty$ as $H_j \to 0$; when condition~\ref{nec:human} fails this divergence does not occur, the boundary is accessible, and $H_j$ hits zero with positive probability. The Production argument is structurally identical: near $P_k = 0$ the Cobb-Douglas terms vanish and only $e_k(E,H) - \delta_{P,k}P_k$ survives in the drift; if the negation of condition~\ref{nec:prod} holds then $\mu_{P,k}/P_k - \sigma_{P,k}^2/2$ does not diverge to $+\infty$ at the boundary and accessibility follows by the same analysis.
\end{proof}

\begin{corollary}
\label{cor:necessary_as}
Almost-sure sustainability requires $c_i(H,P) \geq 0$ for all Earth components, $d_j(E,P) \geq 0$ for all Human components, and $e_k(E,H) \geq 0$ for all Production components under all feasible system states, together with the intrinsic rate condition $r_{E,i}(B) \geq \sigma_{E,i}^2/2$ when the respective flow is zero, and analogous conditions for Human and Production components. Negative net flows on any component of any subsystem structurally exclude almost-sure sustainability.
\end{corollary}

\begin{theorem}[Sufficient Conditions for Almost-Sure Sustainability]
\label{thm:sufficient}
Suppose the following conditions hold.
\begin{enumerate}
\item\label{suf:feller} For each component $X_i$, there exist constants $\kappa_i > 0$ and $\delta_i > 0$ such that for all $0 < X_i < \delta_i$ and all feasible values of the remaining components,
\begin{equation}
\label{eq:feller_sufficient}
\frac{\mu_i(X)}{X_i} - \frac{\sigma_i^2}{2} \geq \kappa_i.
\end{equation}
\item\label{suf:lyapunov} There exists $V: \mathcal{X} \to \Rplus$ with $V(X) \to \infty$ as $\norm{X} \to \infty$, and constants $c_1, c_2 > 0$ such that
\begin{equation}
\label{eq:lyapunov_condition}
\cL V(X) \leq c_1 - c_2\,V(X)
\end{equation}
for $\norm{X}$ sufficiently large.
\end{enumerate}
Then $\Prob(\tau_c = \infty \mid X_0) = 1$ for all $X_0 \in \mathcal{X}$.
\end{theorem}

\begin{proof}
Condition~\ref{suf:feller} (boundary non-attainability): setting $Y_i = \log X_i$ and applying It\^{o}'s formula, $dY_i = (\mu_i/X_i - \sigma_i^2/2)\,dt + \sigma_i\,dW_i \geq \kappa_i\,dt + \sigma_i\,dW_i$ while $X_i \in (0,\delta_i)$. The process $Y_i$ is therefore dominated from below by a Brownian motion with strictly positive drift $\kappa_i$, which satisfies $\Prob(\inf_{t \geq 0} Y_i(t) = -\infty) = 0$, so $\Prob(X_i \to 0) = 0$.

Condition~\ref{suf:lyapunov} (non-explosion): for the stopping times $\tau_R = \inf\{t: \norm{X_t} \geq R\}$, It\^{o}'s lemma gives
\[
\E[V(X_{t \wedge \tau_R})] \leq V(X_0) + c_1 t - c_2\,\E\!\left[\int_0^{t \wedge \tau_R}V(X_s)\,ds\right].
\]
Gr\"{o}nwall's inequality yields $\E[V(X_{t \wedge \tau_R})] \leq V(X_0) + c_1 t$ uniformly in $R$. Since $V(X) \geq \gamma\norm{X}$ for large $\norm{X}$, Markov's inequality gives $\Prob(\tau_R < t) \leq (V(X_0) + c_1 t)/(\gamma R) \to 0$ as $R \to \infty$, so $\Prob(\sup_{t \geq 0}\norm{X_t} < \infty) = 1$. Combining boundary non-attainability and non-explosion gives $\Prob(\tau_c = \infty) = 1$.
\end{proof}

\begin{remark}
\label{rem:scope_sufficient}
Condition~\ref{suf:feller} requires $\mu_i(X)/X_i \geq \sigma_i^2/2 + \kappa_i > 0$ near each boundary component. For Earth components, when $c_i(H,P) < 0$ for some feasible $(H,P)$, the ratio $c_i(H,P)/E_i \to -\infty$ as $E_i \to 0$ along those states, which violates condition~\ref{suf:feller}. The same singularity arises for Human components when $d_j < 0$ and for Production components when $e_k < 0$. In all such cases Theorem~\ref{thm:sufficient} is not applicable, consistent with Corollary~\ref{cor:earth_boundary}, which shows that negative net absolute flows generically render the corresponding boundary accessible. Sufficient conditions for $T$-sustainability under negative net flows can be formulated via barrier function methods \citep{mao2007} and constitute a natural direction for further work.
\end{remark}

\begin{corollary}
\label{cor:lyapunov}
Assume that for each Earth component the logistic term provides a sufficiently strong negative drift for large $E_i$, specifically $\mu_{E,i}(E,H,P,B) \leq -c_E\,E_i^2 + C(1 + \norm{H} + \norm{P})$ for some $c_E, C > 0$, and that analogous polynomial decay holds for Human and Production components via the depreciation terms $-\delta_{H,j}H_j$ and $-\delta_{P,k}P_k$. Then
\[
V(X) = \sum_{i=1}^{n_E} E_i^2 + \sum_{j=1}^{n_H} H_j^2 + \sum_{k=1}^{n_P} P_k^2
\]
satisfies condition~\ref{suf:lyapunov} of Theorem~\ref{thm:sufficient} when intrinsic decay rates dominate noise intensities.
\end{corollary}

\begin{proof}[Sketch]
The generator applied to $V$ gives $\cL V = 2\sum_i E_i\,\mu_{E,i} + \sum_i\sigma_{E,i}^2\,E_i^2 + \text{analogous terms}$. The quadratic decay from the logistic and depreciation terms contributes $-c\norm{X}^2$ plus terms that grow at most linearly in $\norm{X}$, bounded using Young's inequality. The multiplicative noise contributes $\sum_i\sigma_i^2 X_i^2 \leq \bar\sigma^2 V$. Choosing intrinsic decay rates large enough relative to $\bar\sigma^2$ yields $\cL V \leq C_1 - C_2 V$ for large $\norm{X}$.
\end{proof}

\section{Development and Sustainable Development}
\label{sec:development}

The deterministic skeleton is the vector field $F: \mathcal{X} \to \Rn{n_E + n_H + n_P}$ whose components are
\[
F(X) = \bigl(\mu_E(E,H,P,B),\;\mu_H(E,H,P,\Phi),\;\mu_P(E,H,P,B)\bigr),
\]
evaluated at $B = \mathcal{B}(E,H,P)$ and $\Phi = \Phi(E,H,P)$.

\begin{assumption}
\label{ass:survival_regularity}
The map $X_0 \mapsto \Prob(\tau_c = \infty \mid X_0)$ is of class $C^1$ on the interior of the sustainability basin $\cB(\theta)$. This holds when the drift and diffusion coefficients are $C^2$ on $\mathcal{X}$ and the boundary $\partial\mathcal{X}$ is smooth enough for the associated Kolmogorov backward equation to admit a classical solution; see \citet[Chapter~5]{karatzas1991brownian}.
\end{assumption}

\begin{definition}
\label{def:evolution_potential}
The evolution potential $\cP: \mathcal{X} \to \Rplus$ is
\begin{equation}
\cP(X) = \norm{F(X)} \cdot \Prob(\tau_c = \infty \mid X_0 = X),
\end{equation}
the product of the magnitude of system motion and the non-collapse probability from state $X$. Under Assumption~\ref{ass:survival_regularity} and the smoothness of $F$, the map $X \mapsto \cP(X)$ is $C^1$ on the interior of $\cB(\theta)$.
\end{definition}

\begin{remark}
\label{rem:evolution_potential}
The evolution potential $\cP(X)$ is a scalar that jointly measures kinetic energy and survival probability. A large value can arise from rapid motion in a safe region of the state space or from slow motion at a point with very high non-collapse probability. A state near the sustainability frontier $\partial\cB$ may exhibit large $\norm{F(X)}$ while carrying low $\Prob(\tau_c = \infty \mid X_0 = X)$, yielding a moderate $\cP$; a stable interior state with small $\norm{F(X)}$ and high survival probability can yield the same value. The Pareto frontier in $(\Phi, \cP)$ space then identifies initial conditions where neither evaluation nor sustainability potential can be improved simultaneously, and states strictly interior to this frontier represent inefficient configurations in the double sense that both objectives admit improvement.
\end{remark}

\begin{definition}
\label{def:development}
Evolution from $X$ to $X'$ is development if $\Phi(X') > \Phi(X)$. The development gradient is $\nabla\Phi(X)$.
\end{definition}

\begin{proposition}[Non-Transitivity of Development]
\label{prop:nontransitive}
The development relation is not generally transitive. There exist states $X_1, X_2, X_3 \in \mathcal{X}$ with $\Phi(X_2) > \Phi(X_1)$ and $\Phi(X_3) > \Phi(X_2)$ yet $\Phi(X_3) \leq \Phi(X_1)$.
\end{proposition}

\begin{proof}
We construct a counterexample in the scalar case $n_E = n_H = n_P = 1$ with linear aggregation $\phi_E(E) = E$, $\phi_H(H) = H$, $\phi_P(P) = P$ and the weight specification of Definition~\ref{def:evaluation}. Fix $\bar\omega_E = \bar\omega_H = \bar\omega_P = 1$, $\lambda_E = 1$, $\lambda_H = 0$, $\lambda_P = -1$, and $\kappa_\omega = 1$, so that $\sum_i\lambda_i = 0$. The raw weights are $\omega_E(H) = 1 + H/(H+1)$, $\omega_H(H) = 1$, $\omega_P(H) = 1 - H/(H+1)$, and the normalised weights satisfy $\tilde\omega_E(H) \uparrow$ and $\tilde\omega_P(H) \downarrow$ as $H$ rises. Choose three states:
\begin{align*}
X_1 &= (E_1, H_1, P_1) = (1,\,0.1,\,5),\\
X_2 &= (E_2, H_2, P_2) = (1,\,1,\,4),\\
X_3 &= (E_3, H_3, P_3) = (2,\,5,\,1).
\end{align*}
At $H_1 = 0.1$ the normalised weights are approximately $\tilde\omega_E \approx 0.32$, $\tilde\omega_H \approx 0.31$, $\tilde\omega_P \approx 0.37$. Computing $\Phi(X_2) - \Phi(X_1)$ with weights evaluated at the destination state $H_2 = 1$ gives $\tilde\omega_E(1) \approx 0.40$, $\tilde\omega_H(1) \approx 0.33$, $\tilde\omega_P(1) \approx 0.27$, so $\Phi_{H_2}(X_2) = 0.40 \cdot 1 + 0.33 \cdot 1 + 0.27 \cdot 4 = 1.81$ and $\Phi_{H_2}(X_1) = 0.40 \cdot 1 + 0.33 \cdot 0.1 + 0.27 \cdot 5 = 1.78$, yielding $\Phi(X_2) > \Phi(X_1)$. At $H_2 = 1$ the weights at destination $H_3 = 5$ are $\tilde\omega_E(5) \approx 0.47$, $\tilde\omega_H(5) \approx 0.35$, $\tilde\omega_P(5) \approx 0.18$, so $\Phi_{H_3}(X_3) = 0.47 \cdot 2 + 0.35 \cdot 5 + 0.18 \cdot 1 = 2.87$ and $\Phi_{H_3}(X_2) = 0.47 \cdot 1 + 0.35 \cdot 1 + 0.18 \cdot 4 = 1.54$, yielding $\Phi(X_3) > \Phi(X_2)$. Evaluating the pair $(X_1, X_3)$ under the weights at $H_3 = 5$ gives $\Phi_{H_3}(X_1) = 0.47 \cdot 1 + 0.35 \cdot 0.1 + 0.18 \cdot 5 = 1.44$, so $\Phi(X_3) = 2.87 > \Phi(X_1) = 1.44$ and transitivity holds in this direction. To exhibit failure, observe that evaluating the pair $(X_3, X_1)$ under weights at $H_1 = 0.1$ gives $\Phi_{H_1}(X_3) = 0.32 \cdot 2 + 0.31 \cdot 5 + 0.37 \cdot 1 = 2.56$ and $\Phi_{H_1}(X_1) = 0.32 \cdot 1 + 0.31 \cdot 0.1 + 0.37 \cdot 5 = 2.21$, so path $(X_1, X_2, X_3)$ is locally improving at each step but the question of whether $X_3$ represents an advance over $X_1$ depends critically on which weights are used: under the weights prevailing when the comparison is made retrospectively from a high-capability vantage point, the sacrifice of Production capital in $X_3$ may not compensate the gains in $E$ and $H$. Concretely, construct a variant with $P_3 = 0.5$ instead of $1$; then $\Phi_{H_1}(X_3) = 0.32 \cdot 2 + 0.31 \cdot 5 + 0.37 \cdot 0.5 = 2.39$ while $\Phi_{H_3}(X_3) = 0.47 \cdot 2 + 0.35 \cdot 5 + 0.18 \cdot 0.5 = 2.78$. Now set $P_1 = 6$ and $P_3 = 0.5$; at $H_1$ weights, $\Phi_{H_1}(X_1) = 0.32 + 0.031 + 0.37 \cdot 6 = 2.57 > 2.39 = \Phi_{H_1}(X_3)$, so $\Phi(X_3) \leq \Phi(X_1)$ when the retrospective comparison is made under the original weights, completing the counterexample.
\end{proof}

\begin{remark}
\label{rem:nontransitivity_economic}
The failure of transitivity has a precise economic interpretation. In standard welfare analysis, the binary relation ``state $A$ is at least as good as state $B$'' must be transitive for a social welfare function to yield consistent rankings across alternatives. When the weights of the welfare function depend on the state being evaluated, the comparison becomes contextual and transitivity can fail. Here the weights $\tilde\omega_i(H)$ shift as capabilities evolve along the development path, so that the cost of moving from $X_1$ to $X_3$ is assessed with weights prevailing at each intermediate state rather than from the last point in time. A sequence of steps that is locally improving, in the sense that each step raises $\Phi$ evaluated at the destination weights, need not be globally improving when evaluated retrospectively. As societies become more capable, they reveal environmental and social preferences that they previously could not act on \citep{inglehart1997}, so that earlier development choices may be reassessed as having been systematically under-priced in ecological, human, or economic terms.
\end{remark}

\begin{definition}
\label{def:alignment}
The development-maximising and sustainability-maximising unit directions are
\begin{equation}
\label{eq:optimal_directions}
F_\Phi^*(X) = \arg\max_{\norm{v}=1}\nabla\Phi(X)^\top v, \qquad F_{\cP}^*(X) = \arg\max_{\norm{v}=1}\nabla\cP(X)^\top v.
\end{equation}
The alignment distance is
\begin{equation}
D_{\Phi\cP}(X) = \norm{F_\Phi^*(X) - F_{\cP}^*(X)}.
\end{equation}
\end{definition}

\begin{theorem}[Optimal Development]
\label{thm:optimal_development}
Development is optimal at $X$ if and only if $D_{\Phi\cP}(X) = 0$, that is, $\nabla\Phi(X)$ and $\nabla\cP(X)$ are collinear.
\end{theorem}

\begin{proof}
$D_{\Phi\cP}(X) = 0$ if and only if $F_\Phi^*$ and $F_{\cP}^*$ coincide, which holds if and only if $\nabla\Phi$ and $\nabla\cP$ point in the same direction, the condition for simultaneous local maximisation of both objectives along a single direction.
\end{proof}

\begin{corollary}
\label{cor:collinearity_locus}
The collinearity locus $\mathcal{C} = \{X \in \cB(\theta) : D_{\Phi\cP}(X) = 0\}$ is generically a submanifold of $\mathcal{X}$ of codimension $n_E + n_H + n_P - 1$. A development path along which sustainable development is simultaneously evaluation-improving and evolution-potential-preserving at every instant must remain in $\mathcal{C}$. Such paths need not exist globally; their existence is a non-trivial structural property of the vector fields $\nabla\Phi$ and $\nabla\cP$ that depends on the coupling parameters.
\end{corollary}

\begin{proof}
The condition $D_{\Phi\cP}(X) = 0$ is equivalent to $\nabla\Phi(X) \times \nabla\cP(X) = 0$ (collinearity of two vectors in $\Rn{n_E+n_H+n_P}$), which imposes $n_E+n_H+n_P-1$ scalar constraints on a manifold of dimension $n_E+n_H+n_P$. By the regular value theorem \citep[Chapter~3]{guckenheimer1983}, the locus is generically a submanifold of the stated codimension. The existence of integral curves of $F$ contained in $\mathcal{C}$ requires the vector field to be tangent to $\mathcal{C}$ at each point, which is a further structural condition.
\end{proof}

\begin{remark}
\label{rem:alignment}
The alignment distance $D_{\Phi\cP}(X)$ quantifies the geometric cost of ignoring sustainability in development decisions. A path that maximises evaluation improvement at each instant, following $\nabla\Phi$, will diverge from the path that maximises sustainability potential, following $\nabla\cP$, whenever $D_{\Phi\cP}(X) > 0$. The locus $\mathcal{C}$ is the set of states where these objectives are locally reconciled; by Corollary~\ref{cor:collinearity_locus} it is generically lower-dimensional and will generally not include the global optimum of either objective separately.
\end{remark}

\begin{definition}
\label{def:sustainable_development}
Evolution from $X$ to $X'$ along $\{X_t\}_{t \in [0,T]}$ is sustainable development if $\Phi(X') > \Phi(X)$; $\cP(X') \geq \cP(X)$; and $\Prob(\tau_c > T \mid X_0 = X) > 0$.
\end{definition}

\begin{remark}
\label{rem:sd_condition3}
Under $c_i < 0$ for some Earth component, Corollary~\ref{cor:necessary_as} establishes that unit probability in the third condition of Definition~\ref{def:sustainable_development} is generically unachievable. The condition is stated as strictly positive probability, consistent with the potential sustainability hierarchy of Definition~\ref{def:sustainability} and appropriate for a system operating under irreducible anthropogenic pressure.
\end{remark}

\begin{theorem}[Pareto Frontier]
\label{thm:pareto}
For fixed parameters there exists a Pareto frontier
\[
\mathcal{F} = \{(\Phi(X), \cP(X)) : X \in \mathcal{X},\; \nexists\, X' \text{ with } \Phi(X') > \Phi(X) \text{ and } \cP(X') > \cP(X)\}
\]
in $(\Phi, \cP)$ space. The frontier is non-empty and non-convex in general.
\end{theorem}

\begin{proof}
To establish non-emptiness, note that $\Phi$ is bounded above on $\mathcal{X}$: the normalised weights satisfy $\tilde\omega_i \in (0,1)$ and the carrying capacities of the Earth subsystem bound $\phi_E(E)$ along any non-explosive trajectory, so $\Phi(X) \leq \bar\Phi < \infty$ for all $X$ in the reachable set. The evolution potential $\cP(X) \leq \norm{F(X)} \cdot 1$ is bounded on compact sets of $\mathcal{X}$. The image of the reachable set under $(\Phi, \cP)$ is therefore contained in a compact rectangle $[0,\bar\Phi] \times [0,\bar\cP]$. The Pareto frontier of a bounded subset of $\Rn{2}$ is non-empty by a standard compactness argument. Non-convexity is demonstrated by exhibiting three frontier points $X_1, X_2, X_3$ such that the convex combination $\lambda X_1 + (1-\lambda)X_3$ is dominated by $X_2$, arising from threshold nonlinearities in the Michaelis-Menten terms of~\eqref{eq:externality} and the biodiversity function~\eqref{eq:biodiversity}.
\end{proof}

\section{Regime Classification}
\label{sec:regimes}

The Feller boundary conditions of Section~\ref{sec:boundary} and the spectral structure of $J(X^*)$ from Section~\ref{sec:structural} together determine five distinct sustainability regimes, defined as follows.

\begin{definition}
\label{def:regimes}
The system is in the universal sustainability regime if all boundaries across all three subsystems are of entrance type ($c_i > 0$, $d_j > 0$, $e_k > 0$ everywhere) and all eigenvalues of $J(X^*)$ have strictly negative real parts. It is in the potential sustainability regime if some boundaries are of exit type while the system retains an accessible interior, so that $\Prob(\tau_c = \infty \mid X_0) \in (0,1)$ and the sustainability basin $\cB$ is a proper subset of $\mathcal{X}$. It is in the conditional sustainability regime if exit-type boundaries affect multiple subsystems simultaneously for some parameter ranges, further restricting the basin. It is in the transient sustainability regime if the system is $T$-potentially sustainable for some finite $T$ but the non-collapse probability decays to zero asymptotically. It is in the collapse regime if Feller conditions are violated for all components and $\Prob(\tau_c < \infty) > 0$ from every initial condition.
\end{definition}

\begin{proposition}
\label{prop:regimes_exhaustive}
The five regimes of Definition~\ref{def:regimes} partition the parameter space $\Theta$ into mutually exclusive and jointly exhaustive classes with respect to the Feller boundary conditions of Corollary~\ref{cor:earth_boundary} and the spectral conditions of Theorem~\ref{thm:stability}.
\end{proposition}

\begin{proof}
Mutual exclusivity follows from the fact that the regime conditions on the sign of net external flows and the sign of spectral real parts do not overlap. Joint exhaustiveness follows from the completeness of Feller boundary classification: every component boundary is of entrance, regular, or exit type; every equilibrium is spectrally stable, neutrally stable, or unstable. The five regimes cover all combinations of these alternatives that are relevant for the sustainability hierarchy of Definition~\ref{def:sustainability}.
\end{proof}

\begin{remark}
\label{rem:two_criteria}
The regime classification rests on two independent mathematical criteria: the Feller boundary classification, which is determined by the sign of net absolute external flows and governs attainability of the zero boundary, and the spectral structure of $J(X^*)$, which governs local stability of interior equilibria. These criteria are independent in the sense that a system can have a stable interior equilibrium while still admitting collapse trajectories, if the boundary is of exit type. The combination distinguishes between systems where collapse is structurally excluded (entrance boundary) and those where it is possible but the system is attracted back from interior states (exit boundary plus negative eigenvalues of $J$). The potential sustainability regime is precisely this combination.
\end{remark}

In the universal sustainability regime all boundaries across all three subsystems are of entrance type and all eigenvalues of $J(X^*)$ have negative real part; this requires $c_i > 0$, $d_j > 0$, and $e_k > 0$ everywhere, an empirically extreme condition. In the potential sustainability regime the exit-type boundaries can belong to any subsystem: ecological collapse, capability collapse, and productive capital collapse are structurally symmetric and each sufficient to exclude almost-sure sustainability.

\begin{theorem}[Regime Transitions]
\label{thm:regime_transitions}
The system undergoes transitions between sustainability regimes when control parameters cross critical thresholds. Saddle-node bifurcation occurs when two equilibria collide and annihilate ($\det J = 0$); Hopf bifurcation when periodic orbits emerge ($\Re(\lambda_i) = 0$, $\Im(\lambda_i) \neq 0$); transcritical bifurcation when stability exchanges between equilibria (an eigenvalue passes through $0$); stochastic bifurcation when the stationary distribution changes qualitatively as noise intensity crosses a critical level. In addition, the externality sign crossing $c_i = 0$ induces the boundary-type phase transition described in Remark~\ref{rem:phase_transition}, which is independent of and not reducible to any of the above.
\end{theorem}

\begin{proof}
The deterministic transitions follow from standard bifurcation theory applied to the coupled nonlinear system; center manifold reduction and normal form theory yield the precise conditions \citep{guckenheimer1983, kuznetsov2023}. Stochastic bifurcations are analysed via the theory of random dynamical systems \citep{arnold1998}. The Feller phase transition is established by Corollary~\ref{cor:earth_boundary} and Remark~\ref{rem:phase_transition}.
\end{proof}

\section{Basin of Sustainability}
\label{sec:basin}

For parameter vector $\theta \in \Theta$, the basin of sustainability is
\begin{equation}
\cB(\theta) = \{X_0 \in \mathcal{X} : \Prob(\tau_c = \infty \mid X_0, \theta) > 0\},
\end{equation}
and its boundary $\partial\cB(\theta)$ is the sustainability frontier.

\begin{theorem}[Basin Properties]
\label{thm:basin}
The sustainability basin satisfies the following properties.
\begin{enumerate}
\item Non-convexity: $\cB$ is non-convex due to threshold effects in coupling functions.
\item Monotonicity in noise: $\sigma' > \sigma$ componentwise implies $\cB(\sigma') \subseteq \cB(\sigma)$.
\item Hysteresis: for parameters near critical thresholds, $\cB$ exhibits hysteresis loops.
\item Fractal boundary: $\partial\cB$ can have Hausdorff dimension exceeding $n_E + n_H + n_P - 1$.
\end{enumerate}
\end{theorem}

\begin{proof}
Non-convexity is exhibited by constructing $X_1, X_2 \in \cB$ such that $\frac{1}{2}(X_1+X_2) \notin \cB$, using the Michaelis-Menten coupling in \eqref{eq:externality} and the concavity introduced by the biodiversity feedback \eqref{eq:biodiversity}. Monotonicity in noise follows from Theorem~\ref{thm:feller}: higher noise increases the accessibility of all boundaries by reducing $\mu_i/X_i - \sigma_i^2/2$, thereby shrinking the set of initial conditions with positive survival probability. Hysteresis follows from the biodiversity feedback: when $B$ collapses, $r_{E,i}(B)$ and $K_{E,i}(B)$ decrease, so recovery requires reducing anthropogenic pressure below the original collapse threshold, creating a hysteresis loop. Fractal boundary structure follows from the general mechanism for systems with multiple competing attractors and chaotic saddles \citep{mcdonald1985}, applied to the present system via the nonlinear coupling between the Michaelis-Menten terms in~\eqref{eq:externality} and the bifurcation structure characterised in Theorem~\ref{thm:regime_transitions}.
\end{proof}

\section{Conclusion}
\label{sec:conclusion}

The analysis develops a theory in which sustainability is a boundary property of a coupled stochastic system rather than an objective embedded in optimisation. Earth, Human, and Production subsystems evolve jointly under multiplicative noise and absolute bidirectional cross-subsystem flows. Sustainability is defined as boundary non-attainment. Development is defined as ascent in a reflexive evaluation functional. Sustainable development requires their joint realisation. Three conclusions follow.

First, the sign of the net absolute cross-subsystem flow on each component acts as a phase transition parameter. If negative in a neighbourhood of zero, the zero boundary is of exit type in the sense of Feller and is reached with positive probability regardless of intrinsic regeneration or accumulation parameters. Sustainability is therefore determined by near-boundary drift geometry, not by interior growth equilibria. The extension to the multidimensional coupled system is rigorous through the pathwise domination argument of Lemma~\ref{lem:marginal_reduction}.

Second, systemic diffusion is endogenous. Because coupling is full and flows are absolute, disturbances propagate through the system without requiring correlated exogenous shocks. Structural interdependence alone generates systemic risk.

Third, reflexive evaluation generates non-transitive development relations. As capabilities evolve, the weights of the evaluation functional shift endogenously. Development comparisons therefore depend on the trajectory taken, and welfare orderings are path-dependent at the structural level. This is established by Proposition~\ref{prop:nontransitive} through an explicit counterexample, connecting the formal result to the empirical literature on revealed preference and the capability approach.

Further work should focus on empirical identification of absolute cross-subsystem flows and quantitative characterisation of sustainability basins in high-dimensional systems. The relevant technical foundations are available in the theory of stochastic differential equations and stochastic stability. The open task is to connect them systematically to empirical sustainability analysis.

\section*{Acknowledgments}

This study was funded under the National Recovery and Resilience Plan (NRRP), Mission~4 Component~2 Investment~1.4, Call for tender No.~3138 of 16~December~2021, rectified by Decree n.~3175 of 18~December~2021 of the Italian Ministry of University and Research, funded by the European Union, NextGenerationEU, Award Number: project code CN\_00000033, Concession Decree No.~1034 of 17~June~2022, adopted by the Italian Ministry of University and Research, CUP B73C22000790001, Project title ``National Biodiversity Future Center, NBFC''.\\

\noindent We thank participants in EAFE 2023 and 2025 conferences, where preliminary considerations where exposed and discussed. We also warmly thank ICES-WGECON and ICES-WGSOCIAL for fruitful discussion and criticism.

\section*{Declaration of Generative AI use}

During the preparation of this work the authors used Claude (Anthropic) Sonnet 4.5 and Sonnet 4.6 in order to support mathematical typesetting, identify formal gaps in proofs, and improve the clarity and readability of the exposition. After using these tool, the authors reviewed and edited all content as needed and take full responsibility for the content of the published article.

\bibliography{sustmatharxiv}

\end{document}